\documentclass{article}

\usepackage{PRIMEarxiv}

\usepackage[utf8]{inputenc} 
\usepackage[T1]{fontenc}    
\usepackage{hyperref}       
\usepackage{url}            
\usepackage{booktabs}       
\usepackage{amsfonts}       
\usepackage{nicefrac}       
\usepackage{microtype}      
\usepackage{lipsum}
\usepackage{fancyhdr}       
\usepackage{graphicx}       
\graphicspath{{media/}} 
\usepackage[american]{babel}

\usepackage{natbib} 
\usepackage{mathtools} 
\usepackage{booktabs} 
\usepackage{tikz} 
\usepackage{amsmath}
\usepackage{amssymb}
\usepackage{mathtools}
\usepackage{amsthm}

\theoremstyle{plain}
\newtheorem{theorem}{Theorem}[section]
\newtheorem{proposition}[theorem]{Proposition}

\theoremstyle{definition}
\newtheorem{definition}[theorem]{Definition}

\theoremstyle{remark}

\newcommand{\X}{\mathbb{X}}
\newcommand{\Xt}{\mathbb{X}(t)}
\newcommand{\Xtr}{\overline{\X}}
\newcommand{\G}{\mathcal{G}}
\newcommand{\Z}{\mathbb{Z}}
\newcommand{\C}{\mathbf{C}}
\newcommand{\twce}{\Lambda_{i \to j}^{\tau}}
\newcommand{\twcex}[3]{\Lambda_{#1 \to #2}^#3}
\newcommand{\twcee}{\hat{\Lambda}^{\tau}_{i \to j}}

\newcommand{\wvltce}{\Omega(s_i \to s_j)}
\newcommand{\tuamax}{\tau_{\mathrm{max}}}
\newcommand{\Do}{\mathrm{do}}

\newcommand{\Imp}{\mathbf{U}}
\newcommand{\Resp}{\mathbf{V}}
\newcommand{\imp}{\mathbf{u}}
\newcommand{\resp}{\mathbf{v}}
\newcommand{\Ep}{\mathbb{E}}
\newcommand{\Cov}{\mathrm{Cov}}

\title{Causal inference for temporal patterns}

\author{
    Nicolas-Domenic Reiter \\
    German Aerospace Center\\
    Institute of Data Science\\
    07745 Jena, Germany\\
    \texttt{Nicolas-Domenic.Reiter@dlr.de} \\
    \And
    Andreas Gerhardus \\
    German Aerospace Center\\
    Institute of Data Science\\
    07745 Jena, Germany\\
    \texttt{Andreas.Gerhardus@dlr.de} \\
    \And
    Jakob Runge \\
    German Aerospace Center\\
    Institute of Data Science\\
    07745 Jena, Germany\\
    and \\
    Technische Universität Berlin\\
    10623 Berlin, Germany\\
    \texttt{Jakob.Runge@dlr.de} \\
}

\begin{document}
\maketitle

\begin{abstract}	
	Complex dynamical systems are prevalent in many scientific disciplines. In the analysis of such systems two aspects are of particular interest: 1) the temporal patterns along which they evolve and 2) the underlying causal mechanisms. Time-series representations like discrete Fourier and wavelet transforms have been widely applied in order to obtain insights on the temporal structure of complex dynamical systems. Questions of cause and effect can be formalized in the causal inference framework. We propose an elementary and systematic approach to combine time-series representations with causal inference. Our method is based on a notion of causal effects from a cause on an effect process with respect to a pair of temporal patterns. In particular, our framework can be used to study causal effects in the frequency domain. We will see how our approach compares to the well known Granger Causality in the frequency domain. Furthermore, using a singular value decomposition we establish a representation of how one process drives another over a time-window of specified length in terms of temporal impulse-response patterns. To these we will refer to as Causal Orthogonal Functions (COF), a causal analogue of the temporal patterns derived with covariance-based multivariate Singular Spectrum Analysis (mSSA).
\end{abstract}

\section{Introduction}
\label{submission}
Complex dynamical systems are of interest in diverse fields such as neuroscience, climate science, biology, ecology or economics. One approach to gain understanding is to identify and extract spatio-temporal patterns in observational data \cite{townsend2018detection}. There are several data-driven approaches which aim to find such patterns in observational data. Examples are the well-known multivariate Singular Spectrum Analysis (mSSA) \cite{broomhead1986extracting}, \cite{golyandina2020particularities} and Dynamic Mode Decomposition \cite{schmid_2010}. Traditional approaches towards analyzing temporal patterns in time-series data are Fourier and wavelet theory. A discussion of these methods and its applications to climate data is given in \cite{ghil2002advanced} 

The second relevant aspect in the analysis of complex systems is to reveal the underlying causal mechanisms. A variety of techniques have been developed with the aim to uncover causal links form observational data. Some are discussed in \cite{runge2019inferring}. A formal framework to address questions of cause and effect has been developed by \cite{pearl2009causality}. The last decades have shown rapid progress in this direction. A like-wise famous notion of causality is due to \cite{granger1969investigating}. It has been developed for linear stochastic processes in the context of economic time-series analysis. Since its introduction it has seen further advances and has been applied and discussed in many scientific studies \cite{seth2015granger}, \cite{silva2021detecting}. Pearl's framework has been adapted to study causal interactions in discrete-time stochastic processes \cite{eichler2012causal}. The PCMCI Causal Discovery Algorithm \cite{runge2019detecting}, its extensions PCMCI+ \cite{Runge2020DiscoveringCA} and LPCMCI \cite{NEURIPS2020_94e70705} are based on this framework. Under the standard assumptions of constraint-based causal discovery they detect the causal mechanisms governing these processes. This novel method has already been applied in a number of studies, e.g. in the Earth Sciences to analyze causal interaction between the bio-sphere and atmosphere \cite{krich2020estimating} or to study so-called teleconnections in the climate system \cite{di2020dominant}.  

Aforementioned Granger Causality admits a Fourier representation that is commonly referred to as Granger Causality in frequency domain \cite{geweke1982measurement}. A recent approach combines time-series filtering and Granger Causality \cite{faes2017multiscale}. It provides means to analyze causal relations in stochastic processes at multiple time scales. 

Combining stochastic processes with Pearl's causal inference framework opens up a new opportunity towards analyzing emergent causal effects between temporal patterns in such systems.
With this study we seek to conceptually connect time-series decomposition with causal inference in stochastic processes. Here we address the following question: Given a cause and an effect process $\mathbf{X}(t)$ resp. $\mathbf{Y}(t)$, how susceptible is the amplitude of some temporal mode $\resp$ in $\mathbf{Y}(t)$ to perturbations of a temporal mode $\imp$'s amplitude in $\mathbf{X}(t)$? In Section \ref{sec: causal effect representation} we formalize this question as a causal inference problem. Thereby, our approach allows us to define and estimate causal effects with respect to any orthogonal time-series representation. Furthermore, using a singular value decomposition (SVD) we compute pairs of impulse-response modes, which compactly capture the causal influence of $\mathbf{X}(t)$ on $\mathbf{Y}(t)$. We call these pairs Causal Orthogonal Functions (COF). The construction is motivated by the already mentioned mSSA. While mSSA provides a set of combined temporal modes, which are pair-wise uncorrelated and efficiently represent the cross-covariance information of the underlying process, COFs compactly represent the directed influence of one process on another over a time-window of specified length.
In Section \ref{sec: estimation} we comment on estimation of these effects. Section \ref{sec: numerical examples} is devoted to examining several numerical examples. We empirically compare mSSA with COFs and, secondly, contrast Granger Causality in the frequency domain with a notion of frequency-resolved causal effects, which we develop in this work. 

\section{Structural Causal Processes and Causal Inference}\label{sec: causal processes}
\begin{definition}[Structural Causal Process]
	Let $V = \{X^1, \dots, X^{N}\}$ be a set of nodes. A time-series graph (TSG) $\mathcal{G}= (V\times \mathbb{Z}, E)$ is given by $\mathbb{Z}$-indexed copies of $V$ and a set of directed edges whose directions respect the temporal order, i.e. $(X^i_{s}, X^j_t) \in E $ only if $s < t$. The graph $\mathcal{G}$ is called causally stationary if from $(X^i_s, X^j_t) \in E$ it follows that $(X^i_{s+\tau}, X^j_{t+\tau})$ for every $\tau \in \mathbb{Z}$. A Structural Causal Process (SCP) $\X = \{ \Xt\}_{t \in \mathbb{Z}}$ associated to a TSG $\mathcal{G}$ parameterizes each node in $\mathcal{G}$ as a function of its parents in $\mathcal{G}$ and some noise term. In particular, a SCP is specified by a collection of functions $\{f^{i}\}$ and random variables $\{\epsilon^{i}_{t}\}$, where $i \in \{1, \cdots, N\}$ and $t \in \mathbb{Z}$. The random variables $\epsilon_{t}^i$ are assumed to be mutually and serially independent, i.e. $\epsilon_{t}^i$ and $\epsilon_{s}^j$ are independent if $i \neq j$ or $t \neq s$. Let $X^{i}_{t}$ be a node in $\mathcal{G}$, its value is determined by
	\begin{displaymath}
	X^{i}_{t} = f^{i}(\mathbf{pa}(X^{i}), \epsilon^{i}_t)).
	\end{displaymath} 
\end{definition}
If the SCP $\X$ is stationary, then it induces a probability distribution on the set of variables $V \times \mathbb{Z}$. Throughout this work we assume that all SCPs are stationary and that their marginal distributions on finite subsets of $V \times \mathbb{Z}$ have positive densities. Given a SCP $\X$, one can formalize interventions and thereby cause-effect relations between variables in $\X$ \cite{pearl2009causality}. Let $\mathbf{X} = \{X^{i_l}(t_l)\}_{l=1,\dots, L} \subset \mathbb{X}$ and $\mathbf{x}\in \mathbb{R}^L$, then an intervention on $\mathbf{X}$, denoted $\Do (\mathbf{X}) \coloneqq \mathbf{x}$, is defined by a new SCP $\X_{\Do (\mathbf{X}) \coloneqq \mathbf{x}}$ in which $X^{i_l}(t_l) \coloneqq x_l$, i.e. r.v. variable with point-mass distribution $\delta_{x_l}$. All other random variables in the system are defined by the same set of equations and noise variables as in $\X$.
The interventional distribution associated to $\Do(\mathbf{X})\coloneqq \mathbf{x}$ is the distribution of $\X_{\Do (\mathbf{X})\coloneqq \mathbf{x}}$. Graphically, an intervention on $\mathbf{X}$ corresponds to deleting all arrows pointing to any $X^{i_l}(t_l)\in\mathbf{X}$.

\section{Causal Effects and time-series representations}\label{sec: causal effect representation}
Let $\X$ be a SCP and $T$ some positive integer, the length of a time window. Then for each $i \in \{1, \cdots, N\}$ one can define the length $T$ trajectory process
\begin{equation}
\Xtr^i(t, T)
=  
\begin{pmatrix}
X^i(t-T+1) & \cdots & X^{i}(t)
\end{pmatrix}^T
\end{equation} 
where $t \in \Z$. Let us fix some distinct indices $1\leq i,j \leq N$, time-windows $T_i,T_j\geq1$ and some time lag $\tau\geq 0$. Furthermore we pick two orthogonal (or unitary) matrices
\begin{align*}
\Imp &= \begin{pmatrix}
\imp_1 & \cdots & \imp_{T_i}
\end{pmatrix} \in \mathbb{R}^{T_i\times T_i}\\
\Resp &= \begin{pmatrix}
\resp_1 & \cdots & \resp_{T_j}
\end{pmatrix} \in \mathbb{R}^{T_j \times T_j} 
\end{align*}
In practice, these could be associated to time-series decompositions, like a Discrete Wavelet Transform (DWT), Discrete Fourier Transform or Singular Spectrum Analysis (SSA). 
Representing the processes $\Xtr^i$ and $\Xtr^j$ with respect to these decompositions gives 
\begin{align}
\Xtr^i(t - \tau, T_i) &= \sum_{k = 1}^{T_i} \omega^{i}_{k}(t-\tau)\cdot \resp_k\\
\Xtr^j(t, T_j) &= \sum_{k = 1}^{T_j} \omega^{j}_{k}(t)\cdot \imp_k
\end{align} 
Note that each amplitude $\omega^i_k(t) = \langle \imp_k, \Xtr^i(t-\tau, T_i) \rangle$, resp. $\omega^j_l(t) = \langle \resp_l, \Xtr^j(t, T_j) \rangle$ is a random variable and therefore gives rise to a stochastic process. Assuming that the process $\mathbb{X}$ is stationary implies that the corresponding amplitude processes are stationary.
Since amplitudes $\omega^i_k(t-\tau)$ and $\omega^j_l(t)$ are possibly related through the underlying SCP, we address the question of how to define the expected causal effect of $\omega^i_k(t-\tau)$ on $\omega^j_l(t)$. Intuitively, this effect should capture the extent to which the distribution of amplitude $\omega^j_l(t)$ changes if we were to control the amplitude $\omega^i_k(t-\tau)$. These effects would allow us to study lagged causal effects in SCPs with respect to a pair of orthogonal time-series representations.

\begin{definition}[Causal effects between temporal patterns]
	Let $\X$ be a $N$-dimensional SCP which paramterizes a TSG $\G$. Let $\X^i$ and $\X^j$ be two sub-processes of $\X$, $T_i, T_j >0$ two time-window lengths and $\tau \geq0 $ some time lag. Choose some impulse signal $\imp \in \mathbb{R}^{T_i}$ and some response  $\resp \in \mathbb{R}^{T_j}$. The expected causal effect of $\Xtr^i(t-\tau, T_i) \to \Xtr^j(t, T_j)$ with respect to a pair of signals $(\imp, \resp)$ is defined as follows
	\begin{gather*}
	\twce(\mathbf{x})(\imp, \resp) = \\
	\frac{\partial}{\partial \alpha} \Ep(\langle \resp, \Xtr^j(t, T_{j})\rangle|\Do \Xtr^i(t-\tau, T_i) \coloneqq \alpha \imp + \mathbf{x})
	\end{gather*}
	Here $\mathbf{x} \in \mathbb{R}^{T_i}$ is some baseline signal which is perturbed in the above hypothetical experiment by the signal $\imp$.  
\end{definition}
Figure \ref{fig:causal effect time-series representation} graphically illustrates a causal effect with respect to a given pair of signals $(\imp, \resp)$.
\begin{figure}
    \centering{
    \includegraphics[width=8cm]{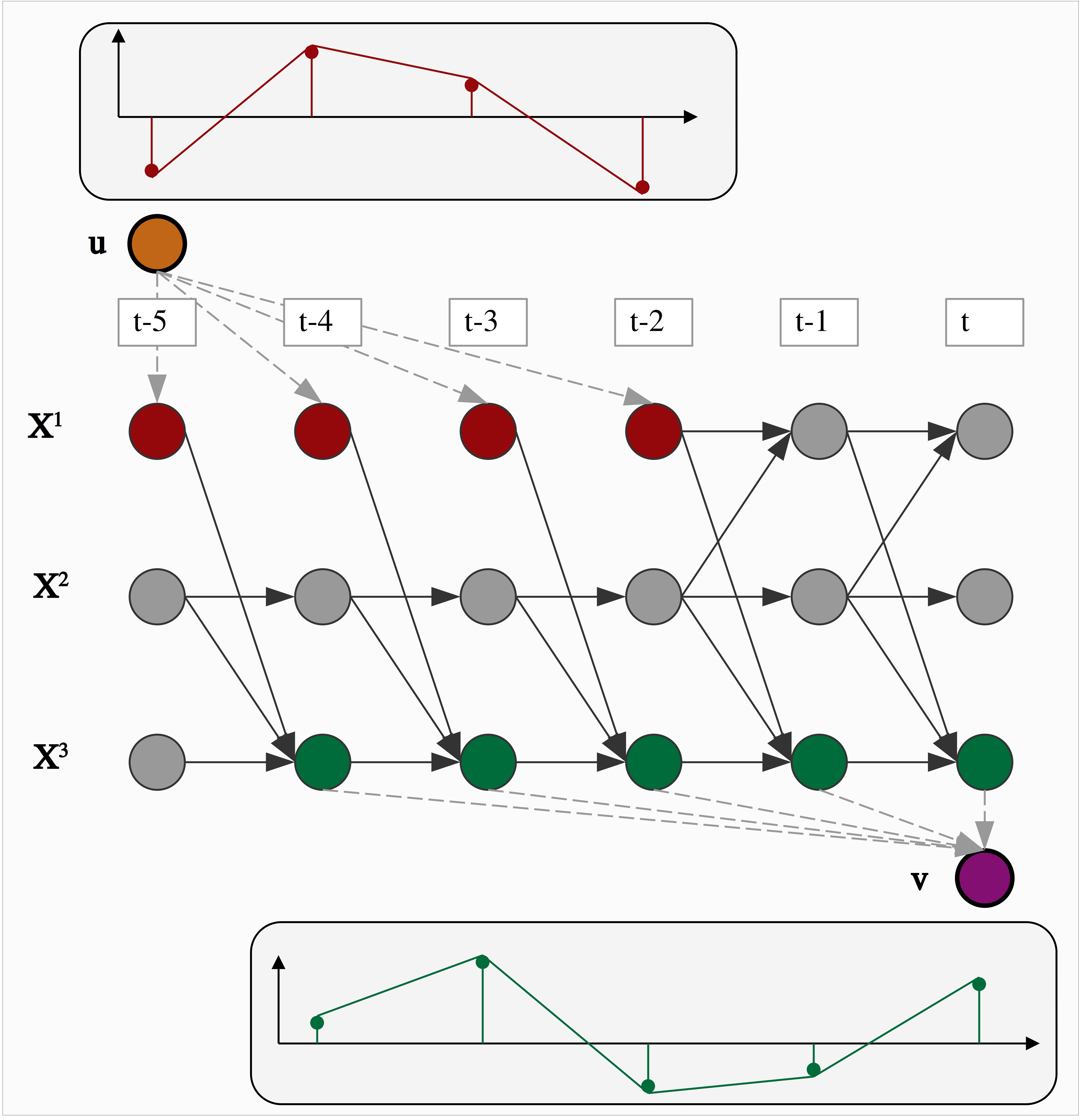}
	\caption{A graphical illustration of the causal effect $\Xtr^1(t-2, 4) \to \Xtr^3(t, 5)$ with respect to the pair of signals $(\imp, \resp)$. Note that there is no arrow pointing to a node in $\Xtr^1(t-2,4)$ (red nodes). That is because we jointly intervene on them. The value of the deterministic orange node represents the amplitude of the pattern $\imp$ (the signal above the time-series graph). In particular, the $k$-th gray-dotted arrow emerging from the orange node is weighted by $\imp_k$. Similarly, the purple colored node stores the amplitude of temporal pattern $\resp$ (represented by the signal below the time-series graph) and the gray-dotted arrows arriving there are defined analogously.}\label{fig:causal effect time-series representation}
    }
\end{figure}

\begin{definition}[Time-windowed causal effect]
	With the same notation as above, one may consider the following function
	\begin{gather*}
	\twce : \mathbb{R}^{T_i} \longrightarrow  \mathbb{R}^{T_j} \\
	\mathbf{x}  \mapsto  \Ep(\Xtr^j(t, T_j)| \Do (\Xtr^i(t-\tau, T_i)) \coloneqq \mathbf{x})
	\end{gather*}
	Assume that $\twce$ is differentiable, thus its Jacobian $D\twce(\mathbf{x})\in \mathbb{R}^{T_j \times T_i}$ can be formed. Its $(l,k)$-th entry is $\twce(\mathbf{x})(\mathbf{e}_k, \mathbf{e}_l)$, i.e. 
	\begin{gather*}
	D\twce(\mathbf{x})_{l,k} = \\
	\frac{\partial}{\partial \alpha} \Ep(X^j(t-T_j+l)|\Do (\Xtr^{i}(t-\tau, T_i)) \coloneqq \alpha \cdot \mathbf{e}_k + \mathbf{x}) 
	\end{gather*}
	In the following we will refer to the matrix $D\twce(\mathbf{x})$ as the time-windowed causal effect matrix (TWCE). It quantifies $\Xtr^i(t-\tau, T_i) \to \Xtr^j(t, T_j)$. 
\end{definition}
A graphical illustration of time-windowed causal effects is given in Figure \ref{fig:time windowed effects}. 
\begin{figure}[ht]
\begin{centering}
    \includegraphics[width=10cm]{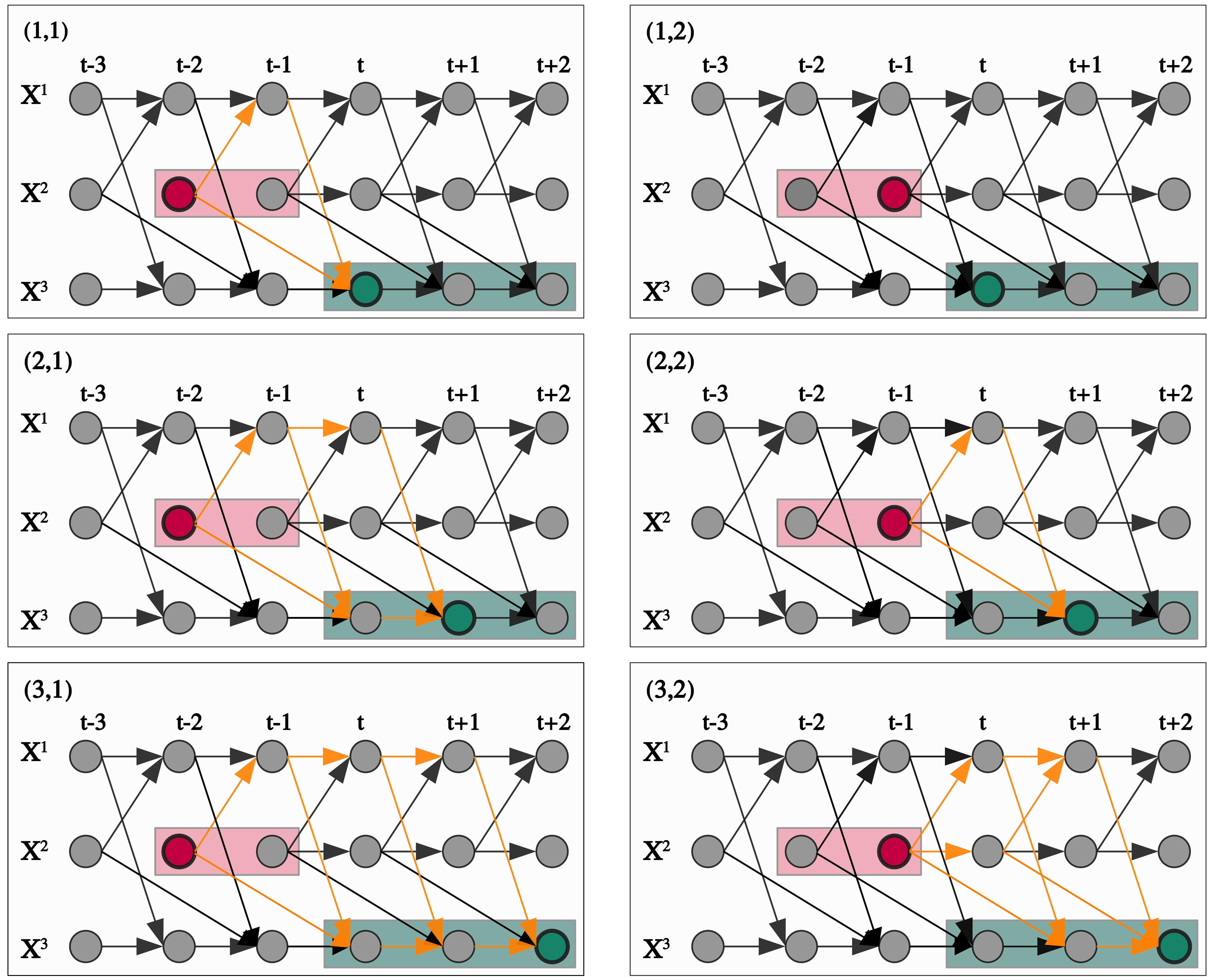}
	\caption{A graphical illustration of the causal effects of which $D\twcex{2}{3}{3}(\mathbf{x})$ is composed of. The graph $(k,l)$ displays which causal effect is quantified in the $l$-th row and $k$-th column of $D\twcex{2}{3}{3}(\mathbf{x})$. Specifically, the red window is the time-window on which we intervene, i.e., on every node lying in this window. Effects of these interventions are measured in the nodes surrounded by the green window. We measure how the expected value of the $l$-th node (green node) in the green window responds to varying the $k$-th component of $\mathbf{x}$ (red node). The effect of this perturbation propagates through the system along the orange arrows.} \label{fig:time windowed effects}
\end{centering}
\end{figure}
The notion of TWCEs is accompanied with another perspective on causal effects with respect to a pair of signals $(\imp, \resp)$. Since both the differential and expectation operator are linear, the causal effect $\Xtr^i(t-\tau, T_i) \to \Xtr^j(t, T_j)$ with respect to time-series representation $\Resp$ on the response side and $\Imp$ on the impulse side admits the matrix factorization
\begin{equation}
(\twce(\mathbf{x})(\imp_k, \resp_l))_{l,k} = \Resp^T D\twce(\mathbf{x}) \Imp.
\end{equation}

\begin{definition}[(Filtered) causal response]
	Given an impulse $\imp \in \mathbb{R}^{T_i}$ and a response signal $\resp\in\mathbb{R}^{T_j}$ both of which normalized to euclidean length one. The causal response (CR) of impulse $\imp$ is $D\twce(\mathbf{x})\imp$. The filtered causal response (FRC) is $\twce(\mathbf{x})(\imp, \resp)\resp$. The causal effect (CE) is the euclidean length of CR. The filtered causal effect (FCE) is the value $\twce(\mathbf{x})(\imp, \resp)$. The relative causal discrepancy ratio
	\begin{equation}\label{eq: relative causal discrepancy}
	CD_{\Delta}(\imp, \resp) = \frac{\lVert \twce(\mathbf{x})(\imp, \resp) \cdot \resp - D\twce(\mathbf{x})\imp \rVert^2}{\lVert D\twce(\mathbf{x}) \imp \rVert^2} \in [0,1]
	\end{equation}
	quantifies how much of the CR of impulse $\imp$ is not captured by the signal $\resp$. On the other hand $1-CD_{\Delta}(\imp, \resp)$ quantifies the squared ratio between FCE and CE.    
\end{definition} 

Time-windowed causal effects allow to phrase any impulse response relation in structural causal processes. For instance, we could fix some impulse signal $\imp$ along which we perturb the cause process away from a base-line trajectory $\mathbf{x}$. On average the intervention will cause the mean interventional trajectory $\Ep(\Xtr^j(t, T_j)|\Do(\Xtr^i(t-\tau,T_i))\coloneqq\mathbf{x})$ to diverge along the CR pattern $D\twce(\mathbf{x})\imp$. Its length, the associated CE quantifies the effect process's sensitivity to the impulse $\imp$ under the intervention $\Do(\Xtr^i(t, T_i)) \coloneqq \mathbf{x}$. 

Or the reverse, along which pattern $\imp$ do we need to perturb the intervention $\Do(\Xtr^i(t, T_i)) \coloneqq \mathbf{x}$ so that $\Ep(\Xtr^j(t, T_j)|\Do(\Xtr^i(t-\tau,T_i))\coloneqq\mathbf{x})$ deviates in the direction of a given signal $\resp$? In general there is no or not a unique answer to this question, since the matrix $D\twce(\mathbf{x})$ might not be invertible. However, using its Moore Penrose pseudo inverse $D\twce(\mathbf{x})^+$ we can define an intervention $\Xtr^j(t, T_j) \coloneqq \mathbf{x}+ \tilde{\imp}$, where $\tilde{\imp} = D\twce(\mathbf{x})^+\resp$ such that: 1.) The distance between its CR and the desired response is minimal, i.e. $\lVert \resp -D\twce(\mathbf{x})\tilde{\imp} \rVert$ and 2.) the absolute value of its FCE, $\lvert \twce(\mathbf{x})(\tilde{\imp}, \resp)\rvert$ is maximal.

\section{Causal Orthogonal Functions}\label{sec: COFs}
In the previous section we provided a possibility to define causal effects with respect to different time-series representations. On the other hand, a SVD on the time-windowed causal effect matrix yields orthogonal time-series representations which efficiently describe the temporal structure of $\Xtr^j(t, T_j)$ driving the trajectory $\Xtr^i(t +\tau, T_i)$
\begin{equation}
\Sigma = \Resp^T D\twce(\mathbf{x}) \Imp.
\end{equation}
The columns of both $\Resp\in\mathbb{R}^{T_j \times T_j}$ and $\Imp \in \mathbb{R}^{T_i \times T_i}$ are orthonormal, and $k$ the rank of $D\twce(\mathbf{x})$. The matrix $\Sigma = diag(\sigma_1, \cdots, \sigma_k) \in \mathbb{R}^{T_i \times T_j}$ carries the singular values of $D\twce(\mathbf{x})$ on its diagonal. They are arranged such that $\sigma_1 \geq \sigma_2 \geq \cdots \geq \sigma_k \geq 0$. Hence, a SVD of the time-windowed causal effect matrix provides us with an orthonormal set of impulse signals, the columns of $\Imp = (\imp_1 \cdots \imp_k)$ and a set of orthonormal response signals, the columns of $\Resp=(\resp_1 \cdots \resp_k)$ such that
\begin{equation}
D\twce(\mathbf{x}) \imp_l = \sigma_l \resp_l = \twce(\imp_l, \resp_l) \resp_l 
\end{equation}
Accordingly, the pairs $\{(\imp_l, \resp_l)\}$ will be referred to as Causal Orthogonal Functions (COFs) and the $l$-th singular value of $D\twce(\mathbf{x})$ is the causal effect associated to the $l$-th pair of COFs. These provide an alternative representation of how one process drives another over time within a complex network of dynamically interacting systems. 

\paragraph{Multi time-scale COFs}
Wavelet Theory and its accompanied Multi Resolution Analysis provides a mathematical framework within which the dynamics of a function or time-series can be examined at a variety of time-scales \cite{percival2000wavelet}. The MRA of a time-series of length $T$, which is assumed to be a multiple of $2^J$, is specified by a so-called wavelet filter and a scaling filter. It gives rise to an orthogonal matrix $\mathcal{W} \in \mathbb{R}^{T\times T}$ which admits the following structure
\begin{displaymath}
\mathcal{W} = \begin{pmatrix}
\mathcal{W}_1 &
\cdots &
\mathcal{W}_J &
\mathcal{V}_J
\end{pmatrix} \in \mathbb{R}^{T\times T}.
\end{displaymath}
The procedure via which this matrix is obtained form the wavelet and scaling filter is described in Chapter 4 of \cite{percival2000wavelet}.The columns of $\mathcal{W}_j \in \mathbb{R}^{T \times T_j}$, where $T_j = \frac{T}{2^j}$, span the $j$-th time-scale. Intuitively, each time-scale characterizes a certain frequency band. The first time scale captures the high-frequency part of the time-series, circa the $[\frac{\pi}{2}, \pi)$ frequency band \cite{percival2000wavelet}. The higher time scales, detect decreasingly lower frequencies, i.e $[\frac{\pi}{2^j}, \frac{\pi}{2^{j-1}} )$. Consequently, the lowest frequency range, approximately $[0, \frac{\pi}{2^J} )$, is spanned by the columns of $\mathcal{V}_J$. If $X\in \mathbb{R}^T$ is a time-series, then its wavelet representation $\mathcal{W}^T X$ exhibits the time-scale specific behavior of $X$, with respect to the underlying wavelet and scaling filter of $\mathcal{W}$. If $\X$ is VAR process and $\twce$ as above, then
\begin{align}
\Omega(s_i \to s_j) &= \mathcal{W}_{s_j}^T D\twce(\mathbf{x}) \mathcal{W}_{s_i},
\end{align}
where $s_j, s_j \in \{1, \dots, J\}$, describes the intra-scale causal effects of $\Xtr^i(t-\tau, T)\to \Xtr^j(t, T)$ at scale $\tau_i$, if $s_i = s_j$, resp. the inter-scale $s_i \to s_j$ causal effects, if $s_i$ and $s_j$ differ. Performing a SVD on $\wvltce$ yields
\begin{equation}
\wvltce = \Resp'(s_i \to s_j)^T \Sigma'(s_i \to s_j) \Imp'(s_i \to s_j)
\end{equation}
Hence, the COFs confined to the intra- or inter-scale interaction are a set of pairs of signals $\{(\imp_k, \resp_k)\}$, where $\imp_k$ is a signal of time-scale $s_i$ and $\resp_k$ a signal of time-scale $s_j$. The effect matrix $D\twce(\mathbf{x})\mathcal{W}_{s_i}$ quantifies the causal influence of $\Xtr^i(t-\tau, T)$ on $\Xtr^j(t, T)$, where the interventions are required to be of time-scale $s_i$. On the other hand, $\mathcal{W}_{\tau_i}^TD\twce(\mathbf{x})$ quantifies the causal influence of $\Xtr^i(t-\tau, T_i)$ on the $s_i$-th time-scale on $\Xtr^j(t, T_j)$.  

\section{Estimation}\label{sec: estimation}
Given that we know the time-series graph $\mathcal{G}$ and the assumptions from section \ref{sec: causal processes}, the proposed causal effects are identifiable, i.e. the effects can be estimated from observational time-series data. This follows from Theorem 4.4.1 \cite{pearl2009causality}. The literature offers different approaches to the estimation of such effects. One option is effect estimation via adjustment sets, see e.g. \cite{perkovic2018complete} for an extensive treatment. It turns out that there are SCPs in which not every time-windowed causal effect can be estimated with this methodology, i.e. time-windowed causal effects are not generally adjustment-identifiable. Problems arise when so-called feedback loops occur, i.e. a causal pathway $X^i(t) \to \cdots \to X^i(t+r)$ which also goes through some other process than $\X^i$. But if the process is a linear VAR process, we have another option to conduct the discussed analysis of COFs. Then it suffices to estimate the VAR process and compute the time-windowed causal effect matrix using \ref{prop: linear effects var}. In particular, one then does not have to worry about aforementioned feedback loops, at least theoretically. Estimation of time-windowed causal effect matrices in VAR processes then involves three steps: 1.) we estimate its TSG $\hat{\mathcal{G}}$ (e.g. using PCMCI \cite{runge2019detecting}), based on which 2.) the coefficient of each link in $\hat{\mathcal{G}}$ is estimated, so that $\hat{\Phi}$ is our estimated VAR model. Finally, we evaluate the analytic expression described in \ref{prop: linear effects var} on the estimated VAR model $\hat{\Phi}$. This gives the estimate $D\twcee(\mathbf{0})$ on which we then perform a SVD in order to determine the pairs $\{(\hat{u}_k, \hat{v}_k)\}$ of COFs with associated causal effects $\{\hat{\sigma}_k\}$. The estimated impulse-response pairs could be assessed by computing the causal discrepancy measure $CD_{\Delta}(\hat{\imp}, \hat{\resp})$ on them, using the true effect matrix.  

\section{Numerical Examples}\label{sec: numerical examples}
In this section we illustrate time-windowed causal effects by examining a number of VAR processes. First, we empirically compare the temporal modes extracted by mSSA with the COF impulse response pairs. Secondly, a comparison between Granger Causality in the Frequency Domain and time-windowed causal effects in the frequency domain will be conducted. Recall, a VAR process is modeled as follows
\begin{equation}
\X(t) = \sum_{p =1}^{\tau_{\mathrm{max}}} \Phi(p)\X(t-p) + \eta(t).
\end{equation}  
\begin{figure}[ht]
\begin{centering}
    \includegraphics[width=10cm]{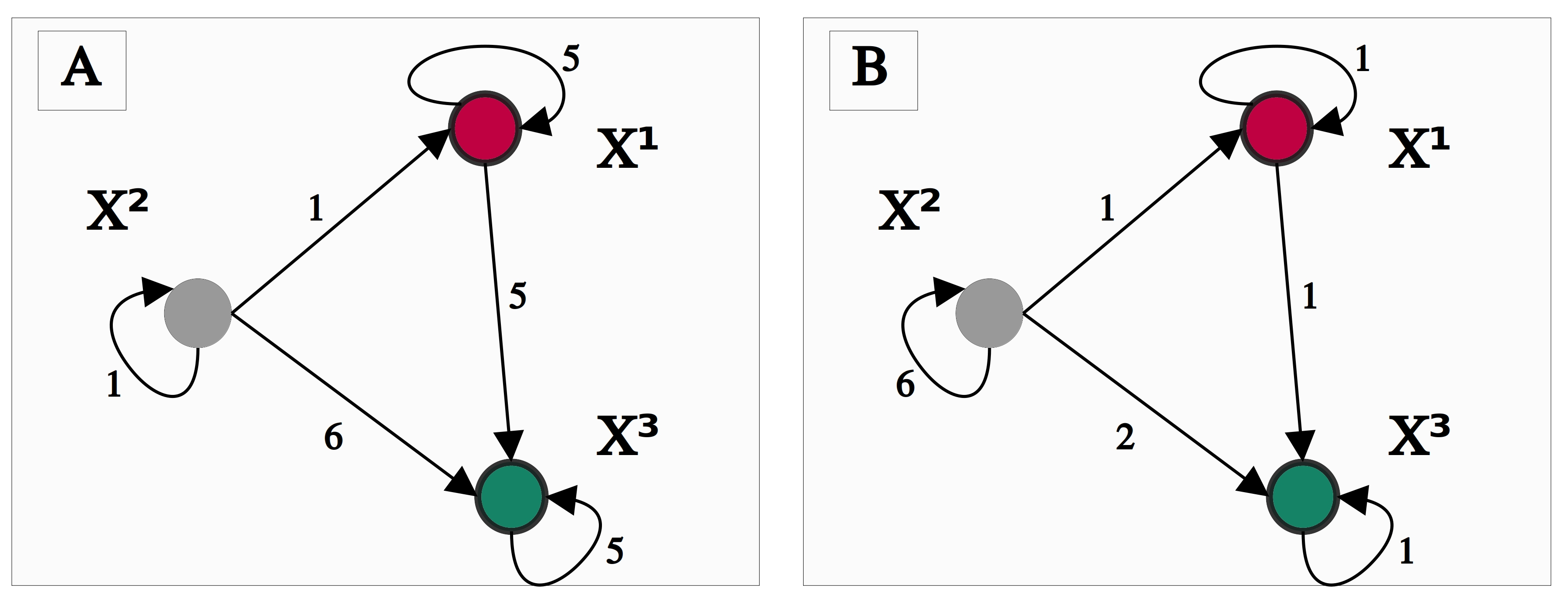}
	\caption{The summary graphs of two VAR processes are shown. The parameters for process $\mathbf{A}$ are chosen as follows: $\Phi^A(5)_{1,1}=\Phi^A(1)_{2,2}=\Phi^A(5)_{3,3} = 0.8$, $\Phi^A(5)_{3,1}=-0.3$, $\Phi^A(1)_{1,2} = \Phi^A(6)_{3,2}=0.7$. The process corresponding to graph $\mathbf{B}$ is defined by the following parameters: $\Phi^B(1)_{1,1} = \Phi^B(1)_{3,3} = \Phi^B(6)_{2,2} = 0.8, \Phi^B(1)_{1,2}(1) = \Phi^B(2)_{3,2}=0.7, \Phi^B(1)_{3,1}=-0.3$}
	\label{fig:example_var}
\end{centering}
\end{figure}
Each $\Phi(p)$ is a $N \times N$-matrix and encodes the lag $p$ dependencies in $\X$, $\eta(t)$ is a zero-mean Gaussian random vector, whose covariance is the identity matrix $\mathbb{I} \in \mathbb{R}^{N \times N}$. Each VAR process has an associated summary graph. Its vertices are the variables $\X$. There is an edge $\X^i \to \X^j$ annotated with number $p$ if the matrix entry $\Phi(p)_{j,i}$ is non-zero.

If $\X$ is VAR process, then the time-windowed causal effect matrix $D\twce(\mathbf{x})$ is independent of $\mathbf{x}$. Hence, in this section time-windowed causal effect matrices will simply denoted $\twce$.

For both of the VAR processes defined in Figure \ref{fig:example_var} the time-windowed causal effect matrix $\twcex{3}{1}{0}$ for time-window length $T=100$ is computed according to \ref{prop: linear effects var} in the Appendix. The coefficients are such that those processes are stable and thus stationary. Finally, on each of these matrices a SVD is performed, based on which the COFs are defined. In Figure \ref{fig: mssa_vs_cofs_1} and \ref{fig: mssa_vs_cofs} the first three COF modes are plotted for the VAR process $\mathbf{A}$, resp. $\mathbf{B}$ as given in Figure \ref{fig:example_var}. 
  
\subsection{Comparing MSSA modes with COF pairs}

Causal Orthogonal Functions are by construction similar to the well known mSSA. For some time-window $T$, mSSA extracts combined temporal patterns from the lagged covariance matrix $\mathbf{C}=(\mathbf{C}_{i,j})_{1\leq i,j\leq N}\in \mathbb{R}^{T\cdot N \times T\cdot N}$, where $\mathbf{C}_{i,j} = (\Cov(X^i(t-T+r), X^j(t-T+s)))_{1\leq r,s\leq T}\in\mathbb{R}^{T\times T}$. The extracted patterns arise from diagonalization, $\mathbf{D} = \mathbf{E}^T \mathbf{C} \mathbf{E}$, thus $\mathbf{D}$ is diagonal having only real non-negative entries and $\mathbf{E}$ is orthogonal. Its $k$-th column is a tuple of temporal modes $(\mathbf{e}_{1,k}, \cdots, \mathbf{e}_{N, k})$, where $\mathbf{e}_{j,k}\in\mathbb{R}^T$ expresses a temporal pattern in the process $\X^j$ over time-period $T$. The columns are assumed to be ordered such that the diagonal entries of $\mathbf{D}$ are non-decreasing. Hence, the combined patterns $\mathbf{e}_k$ and $\mathbf{e}_l$ are uncorrelated with respect to $\mathbf{C}$ if $k \neq l$ and otherwise their covariance with respect to $\mathbf{C}$ is given by $\mathbf{D}_{k,k}$. In particular, the first combined patterns are those which explain most of the cross-covariance information $\mathbf{C}$. If $\X$ is a VAR process, then the entries in $\mathbf{C}$ arise from the coefficients defining the VAR process. They quantify the association between $\X^i(t+r)$ and $\X^j(t)$. From a graphical point of view, all causal and confounding paths, connecting $\X^i$ and $\X^j$ are weighted by the coefficients of $\X$. On the other hand, the time-windowed causal effect matrix $\twce$ accumulates all causal paths, weighted by the coefficients of $\X$, going from $\X^i$ into and through $\X^j$ over time-window $T$, see Proposition \ref{prop: linear effects var} in th Appendix.

Motivated by the conceptual similarity between COF and mSSA, we ask the following question: Do the combined mSSA modes implicitly encode the process's causal-temporal structure? Specifically, for the time-windowed causal effect $\Xtr^i(t,T)\to\Xtr^j(t,T)$ and impulse $\Do\Xtr^i(t,T)\coloneqq\mathbf{e}_{i,k}$, how much of its response is explained by the pattern $\mathbf{e}_{j,k}$? 

In the following we illustrate the difference between mSSA and COF modes on the examples defined in Figure \ref{fig:example_var}. For each process we generate 10000 data points and estimate the MSSA according to \cite{ghil2002advanced} for time-window $T=100$. The $k$-th combined mSSA mode will be denoted $(\mathbf{e}_{1,k}, \mathbf{e}_{2,k}, \mathbf{e}_{3, k})$, where $1\leq k \leq 3 $. We scale each $\mathbf{e}_{i,k}$ to euclidean length one. With our framework we study the time-windowed causal effect $\Xtr^1(t,100) \to \Xtr^3(t,100)$. Specifically, we investigate what temporal patterns in $\Xtr^3(t,100)$ emerge from MSSA-based impulses $\Do \Xtr^1(t,100)\coloneqq\mathbf{e}_{1,k}$. Three aspects of the responses caused by these impulses will be examined. First, the CR pattern $\twcex{1}{3}{0}\mathbf{e}_{1,k}$ and its associated CE $\mu_k = \lVert\twcex{1}{3}{0}\mathbf{e}_{1,k}\rVert$. Secondly, the FCE $\lambda_k = \twcex{1}{3}{0}(\mathbf{e}_{1,k}, \mathbf{e}_{3,k})$ by which impulse $\mathbf{e}_{1,k}$ amplifies pattern $\mathbf{e}_{3,k}$ in $\Xtr^3(t,100)$ and finally, the causal discrepancy measure $CD_{\Delta}(\mathbf{e}_{1,k}, \mathbf{e}_{3,k})$. The last two aspects indicate to what degree the pairs $(\mathbf{e}_{1,k}, \mathbf{e}_{3,k})$ capture the causal-temporal structure of $\Xtr^1(t,100)\to \Xtr^3(t, 100)$.

In process $\mathbf{A}$ (see Figure \ref{fig:example_var}), the auto-lag structure is chosen such that both $\X^1$ and $\X^3$ generate high-frequency signals, while $\X^2$ produces a low-frequency signal. Hence $\mathbf{A}$ describes a system where the causal link $\X^1 \to \X^3$ between high-frequency process is confounded by a low-frequency process $\X^2$. The first three COF modes (the blue signals in first and second column of Figure \ref{fig: mssa_vs_cofs}) reflect $\X^3$'s sensitivity to high-frequency interventions in $\X^1$. However, the first three mSSA modes (the orange signals in first and second column of Figure \ref{fig: mssa_vs_cofs}) are dominated by the low-frequency confounding $\X^2$. Figure \ref{fig: mssa_vs_cofs} shows that the second and third pairs of MSSA modes considerably miss the causal-temporal structure. Whereas the causal discrepancy measures of these pairs express this numerically, the orange signals in the second column of Figure \ref{fig: mssa_vs_cofs}  illustrate this visually.    

An analogous analysis (see Figure \ref{fig: mssa_vs_cofs_1}) was performed on the VAR process defined by the graph $\mathbf{B}$ in Figure \ref{fig:example_var}, which is conceptually similar to $\mathbf{A}$. The difference is that the auto-lag structure of cause and effect process are such that they generate low-frequency signals, while the auto-lag structure of the confounding process generates a high-frequency signal. Here we make similar observations as in the previous example. For two of the three MSSA pairs the expected projected effects are profoundly smaller than the effects associated to the pairs of COF modes and yield substantial causal discrepancy numbers. On the other hand, the third pair of mSSA modes well represents the temporal structure of causal influence $\X^1 \to \X^3$ as the causal discrepancy value $0.007$ indicates that pattern $\mathbf{e}_{3,3}$ is an accurate approximation for the response of impulse $\mathbf{e}_{1,3}$.   
\begin{figure}[hbt!]
\begin{centering}
    \includegraphics[width=8cm]{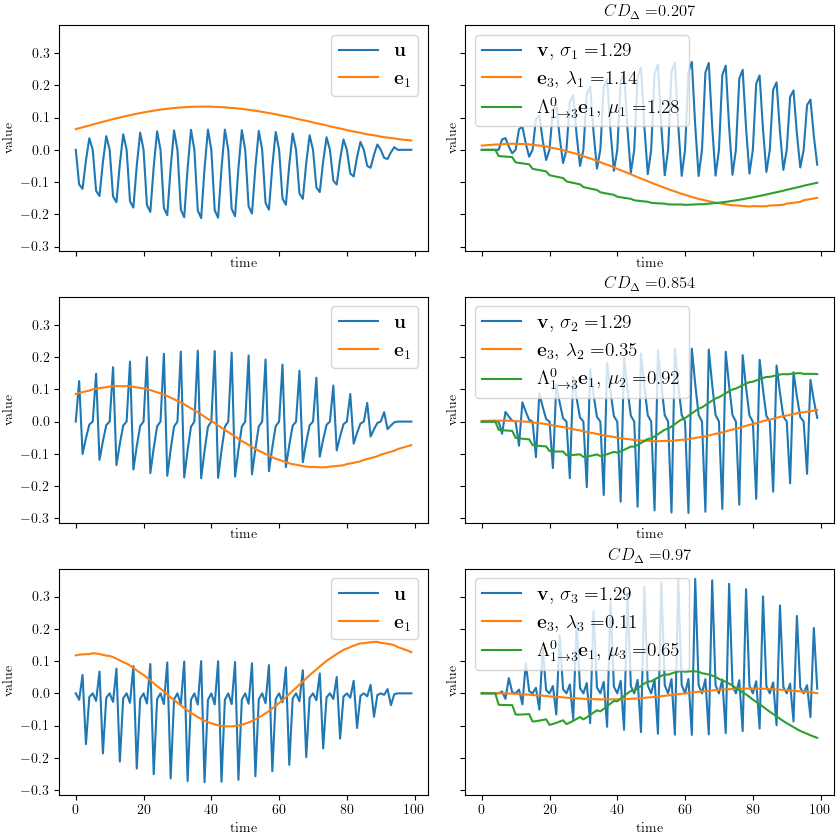}
	\caption{Each row in the plot corresponds to a pair of COF and MSSA modes for the process $\mathbf{A}$ from Figure \ref{fig:example_var}. First column: The $k$-th COF impulse $\imp_k$ (blue signal), first component $\mathbf{e}_{1,k}$ of $k$-th MSSA mode. Second column: the causal response pattern $\sigma_k \cdot\resp_k$ (blue signal) of impulse $\imp_k$, causal response $\twcex{1}{3}{0}\mathbf{e}_{1,k}$ (green signal) and $\mu_k = \lVert\twcex{1}{3}{0}\mathbf{e}_{1,k}\rVert$, the amplification of $\mathbf{e}_{3,k}$ (orange signal) due to impulse $\mathbf{e}_{1,k}$ and $\lambda_k = \langle\twcex{1}{0}{3}\mathbf{e}_{1,k}, \mathbf{e}_{3,k}\rangle$}
	\label{fig: mssa_vs_cofs}
\end{centering}
\end{figure} 
    
\begin{figure}[hbt!]
    \begin{centering}
    \includegraphics[width=8cm]{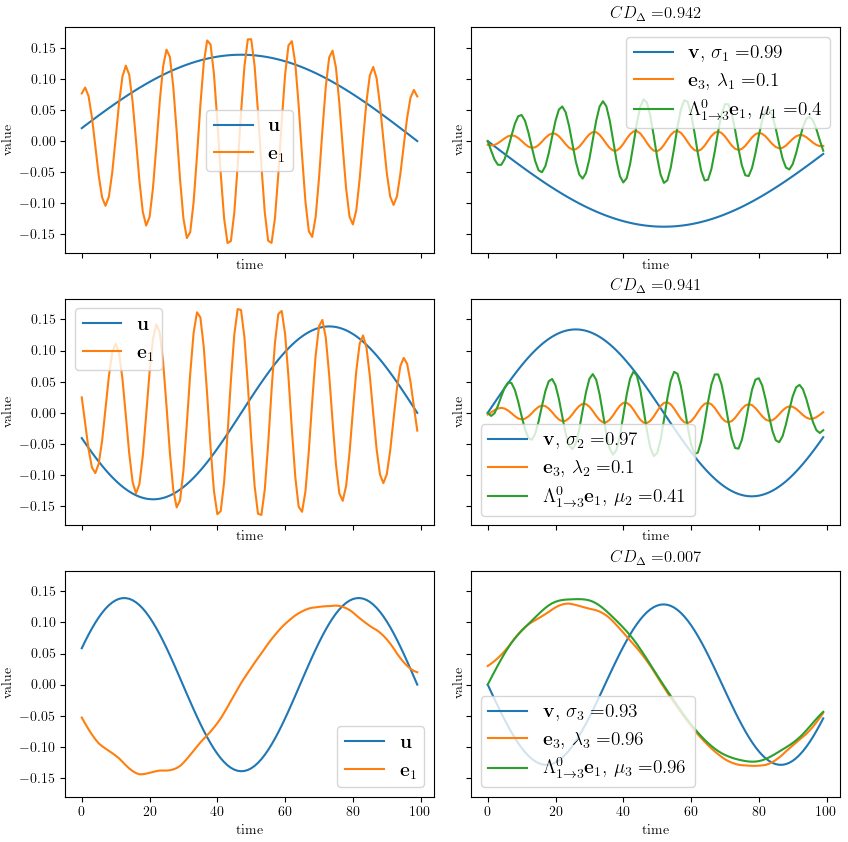}
	\caption{Analogous to the plots in Figure \ref{fig: mssa_vs_cofs_1} for the VAR process defined by $\mathbf{B}$ in Figure \ref{fig:example_var}}
	\label{fig: mssa_vs_cofs_1}
    \end{centering}
\end{figure}      

\subsection{Comparing the Fourier representation of time-windowed effects with Frequency Granger Causality}\label{sec: comparison FGC}
In this section we study time-windowed causal effects of the form $\Xtr^i(t, T)\to \Xtr^j(t, T)$, i.e. the time-windows for cause and effect process have the same size and there is no time-lag between them ($\tau = 0$). Based on the Fourier representation of its associated effect matrix $\twcex{i}{j}{0}$ we define causal effects in the frequency domain, which we then compare with frequency-based Granger Causality on a selection of bi-variate VAR processes. 

Let $\mathcal{F} \in \mathbb{C}^{T\times T}$ be the Discrete Fourier Transform (DFT) matrix and $\mathcal{F}^H$ its complex conjugate transpose. Then $\frac{1}{\sqrt{T}}\mathcal{F}$ is unitary and therefore $\frac{1}{T}\mathcal{F}\twcex{j}{i}{0}\mathcal{F}^H$ is the Fourier time-windowed causal effect matrix (FTWC). It is composed of the complex valued effects with respect to the Fourier basis. We proceed by taking entry-wise absolute values of FTWC. The $k$-th element on its diagonal is the causal effect of $\X^i$ on $\X^j$ at frequency $\omega_k = 2\pi\frac{k}{T}$, where $k\in\{0, \dots, T'\}$ and $T' = \frac{T}{2}$ if $T$ even or $T' = \frac{T-1}{2}$ otherwise. The effect at $\omega_k$ will be denoted $\mathcal{F}^{CE}_{i \to j}(\omega_k)$. It measures the average factor by which an intervention on $\X^i$ at frequency $\omega_k$ amplifies the frequency $\omega_k$ in $\X^j$. Plotting the values of $\mathcal{F}^{CE}_{i\to j}$ thus illustrates the frequency-resolved response structure of $\X^j$ with respect to interventions on $\X^i$. By construction, this does not take into account the internal dynamics of the cause process $\X^i$ (because the causal effect involves an intervention on the cause process).

Granger Causality in the frequency domain \cite{geweke1982measurement} is the classical frequency-based causality measure for VAR processes. To each frequency $0\leq\omega\leq\pi$ it associates a quantity $\mathcal{F}^{GC}_{i \to j}(\omega)$. In the following we evaluate $\mathcal{F}^{GC}$ on the same set of frequencies on which $\mathcal{F}^{CE}$ is defined and compare the plotted values.  

We will now investigate six bi-variate processes to demonstrate the difference between $\mathcal{F}^{GC}$ and $\mathcal{F}^{CE}$, where time-window length is set to $T=200$. The VAR processes we investigate parameterize the structural causal process graphs displayed in Figure \ref{fig: gc_vs_frequency_effects_vars}. We applied $\mathcal{F}^{GC}$ and $\mathcal{F}^{CE}$ to all of these processes and display the results in Figure \ref{fig: gc_vs_frequency_effects_values}.

The VAR processes corresponding to the graphs $\mathbf{A.1}$ and $\mathbf{A.2}$ differ form each other in the effect process $\X^2$ dynamics. The frequency effect measure $\mathcal{F}^{CE}$ is able to distinguish these two VAR processes. In contrast, $\mathcal{F}^{GC}$ is not able to tell them apart (see plot in first row of Figure \ref{fig: gc_vs_frequency_effects_values}). The VAR processes corresponding to the graphs $\mathbf{B.1}$ and  $\mathbf{B.2}$ distinguish from one another only in the dynamics of the cause process $\X^1$. The corresponding Frequency Granger Causality measures $\mathcal{F}^{GC}$ reflects the difference between the dynamics in $\X^1$, while $\mathcal{F}^{CE}$, as noted above, does not (see plot in second row of Figure \ref{fig: gc_vs_frequency_effects_values}). Finally, each of the VAR processes parameterized by the graphs $\mathbf{C.1}$ and $\mathbf{C.2}$ have the same dynamics in the cause and effect process. On these processes both measures $\mathcal{F}^{GC}$ and $\mathcal{F}^{CE}$ behave similarly (third row in Figure \ref{fig: gc_vs_frequency_effects_values}). 

Figure \ref{fig: gc_vs_frequency_effects_vars} and \ref{fig: gc_vs_frequency_effects_values} can be summarized as follows: If there is a causal relation form a cause to an effect process, the frequency based GC detects the cause process's dynamics, whereas frequency causal effects identify the expected frequency-resolved responsiveness of the effect process to interventions on the cause process.        
\begin{figure}[hbt!]
    \begin{centering}
    \includegraphics[width=8cm]{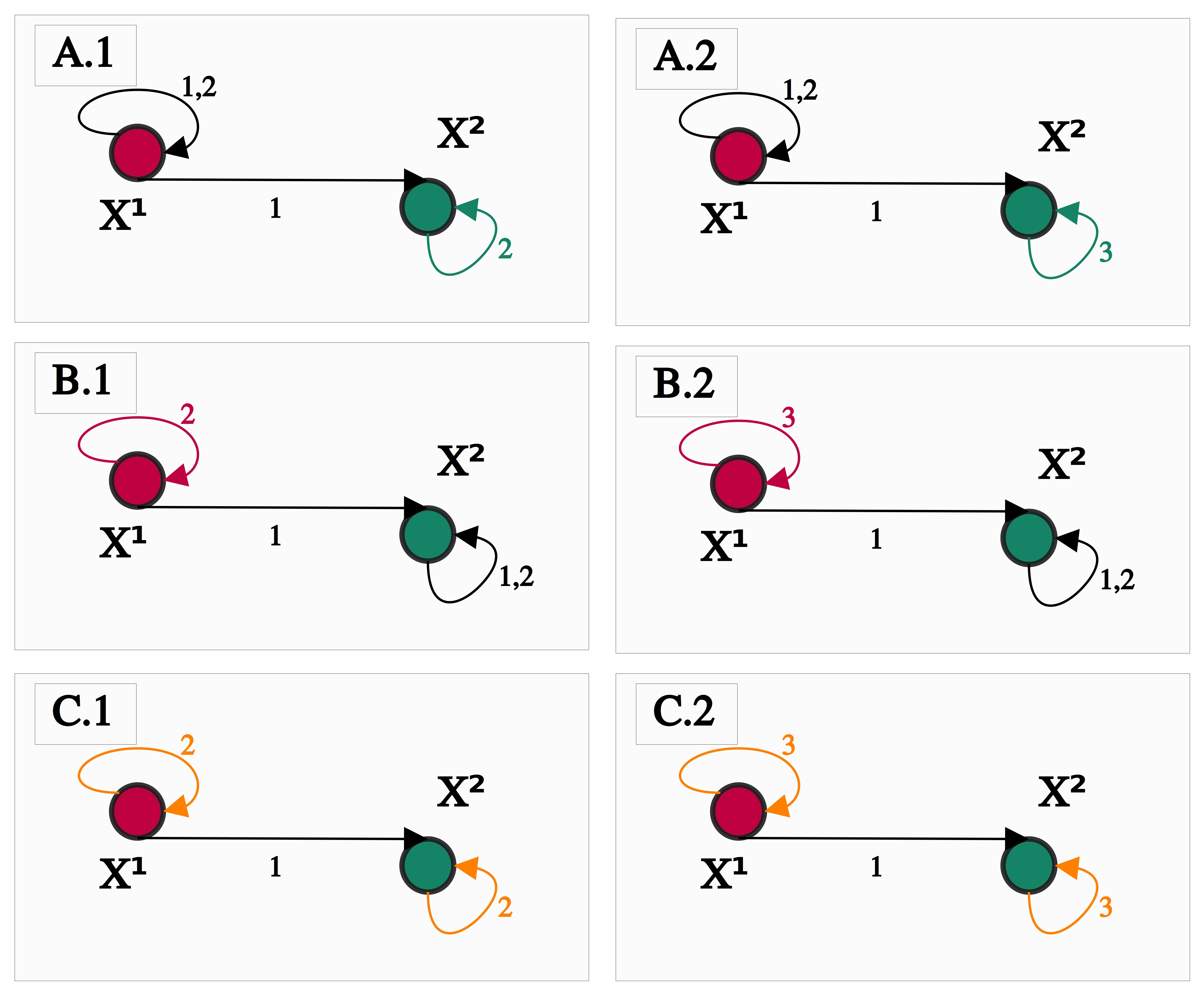}
	\caption{Six summary graphs of parameterized VAR processes. Parameters for $\mathbf{A.1}$ and $\mathbf{A.2}$: $\Phi_{\mathbf{A.1}}(1)_{1,1} = \Phi_{\mathbf{A.2}}(1)_{1,1} = 0.5$, $\Phi_{\mathbf{A.1}}(2)_{1,1} = \Phi_{\mathbf{A.2}}(2)_{1,1} = -0.8$, $\Phi_{\mathbf{A.1}}(1)_{2,1}=\Phi_{\mathbf{A.2}}(1)_{2,1} =0.25$, $\Phi_{\mathbf{A.1}}(2)_{2,2} = \Phi_{\mathbf{A.2}}(3)_{2,2} = -0.8$; Parameters for $\mathbf{B.1}$ and $\mathbf{B.2}$: $\Phi_{\mathbf{B.1}}(2)_{1,1} = \Phi_{\mathbf{B.2}}(3)_{1,1} = -0.8$, $\Phi_{\mathbf{B.1}}(1)_{2,1} = \Phi_{\mathbf{B.2}}(1)_{2,1} = 0.25$, $\Phi_{\mathbf{B.1}}(1)_{2,2} = \Phi_{\mathbf{B.2}}(1)=0.55$, $\Phi_{\mathbf{B.1}}(2)_{2,2} = \Phi_{\mathbf{B.2}}(2)_{2,2} = -0.8$; Parameters for $\mathbf{C.1}$ and $\mathbf{C.2}$: $\Phi_{\mathbf{C.1}}(2)_{1,1} = \Phi_{\mathbf{C.1}}(2)_{2,2} = -0.8$, $\Phi_{\mathbf{C.2}}(1)_{2,1} = \Phi_{\mathbf{C.2}}(1)_{2,1} = 0.25$, $\Phi_{\mathbf{C.1}}(1)_{1,1} = \Phi_{\mathbf{C.2}}(1)_{2,2} = 0.55$}
	\label{fig: gc_vs_frequency_effects_vars}
    \end{centering}
\end{figure} 
\begin{figure}[hbt!]
    \begin{centering}
    \includegraphics[width=8cm]{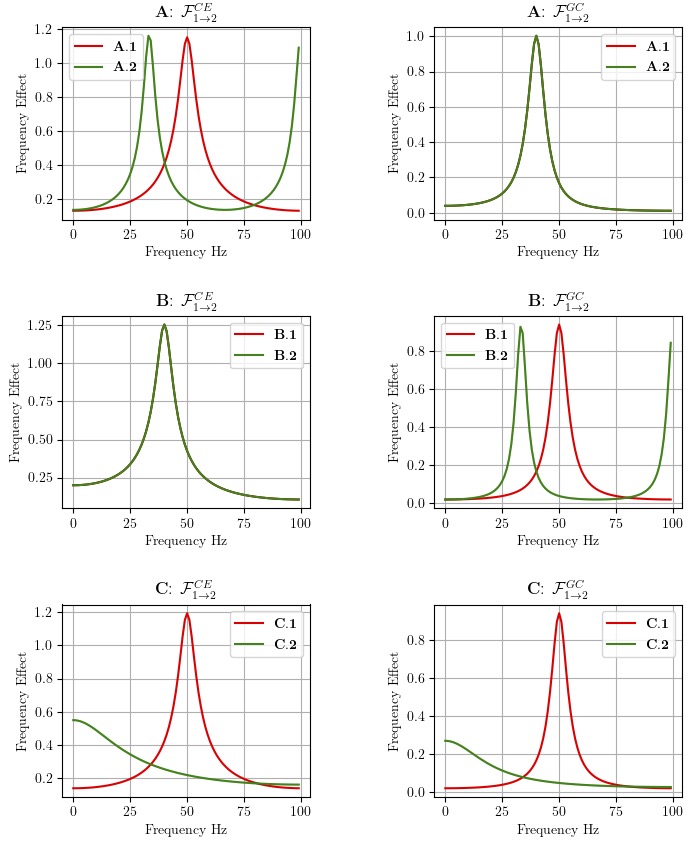}
	\caption{Left column: Values of frequency causal effects $\mathcal{F}^{CE}_{1\to 2}$ are plotted for each VAR process from Figure \ref{fig: gc_vs_frequency_effects_vars}, where the labels of the plots correspond to the indexing of the VAR processes. Right column: Analogous for frequency-based Granger Causality $\mathcal{F}^{GC}_{1 \to 2}$}
	\label{fig: gc_vs_frequency_effects_values}
    \end{centering}
\end{figure}

\section{Outlook and Discussion}
In this study we used the notion of interventions, a well established concept in causal inference, in order to define how the trajectory of a cause process influences the trajectory an effect process. Via this approach causal effects with respect to different time-series representations can be studied. This provides means to examine the causal-temporal structure in SCPs. Furthermore, we derived a representation of how one process drives another over a specified time-window. This representation compactly captures how information propagates from a cause process into and through an effect process. COFs are conceptually related to the well-known mSSA method in that COFs aim to give an efficient representation of similar linear structures, i.e the time-windowed causal effect matrix vs. the cross-covariance structure. We empirically investigated to what degree the first modes generated with the mSSA method identify the temporal structure of causal information flow. For that we considered two example VAR processes, in which the causal link of interest was confounded. The result of our analysis shows that even the strongest mSSA modes may struggle to extract patterns along which information flows from one process into and through another. Furthermore, as mSSA is purely covariance-based no directional information of influence can be extracted.

We used the Fourier representation of the time-windowed causal effect matrices to quantify the causal influence of one process on another in the frequency domain. Plotting this measure for a given cause and effect process illustrates the response structure of the latter, i.e. it indicates on which frequencies one has to intervene in the cause process in order to realize a pronounced response in the effect process. Note the clear interpretability of these effects. On the contrary, there has been debate about the "causal" interpretation of Granger Causality in the frequency domain \cite{stokes2017study}, \cite{barnett2018solved}, \cite{faes2017interpretability}, \cite{dhamala2018granger}. 
Finally, we demonstrated qualitative differences and similarities to GC in frequency domain on six different bivariate VAR processes. Our comparison suggests that the proposed frequency causal effect measure distinguishes from GC as follows: 1.) Frequency based causal effects are sensitive to the auto-lag structure in the effect process but indifferent to the auto-lag structure in the cause process, while 2.) the dynamics of the cause but not of the effect process are reflected by Frequency GC. If both cause and effect process have the same auto-lag structure, then the plots of GC and frequency effects have similar shapes.  

The methodology we are proposing distinguishes from \cite{faes2017multiscale}, where a VAR process is filtered and subsequently down-sampled producing a new process, which admits a representation as a so-called innovation state space model. Then a Granger Causality analysis, based on the GC concept developed in \cite{solo2016state}, is conducted on this process. Contrary to that, our proposed method does not aggregate on the process level, instead causal effects are aggregated, i.e. we are using the filters 1.) to define interventions and 2.) quantify their expected effects.     

Regarding assumptions relevant for causal inference, we excluded the possibility of latent processes or contemporaneous links. Studying the proposed effects after relaxing these conditions could be a matter of subsequent work. Testing the proposed estimator of time-windowed causal effect matrices for VAR processes could be subject of a systematical and detailed simulation study. This could include a comparison with different approaches towards estimating time-windowed causal effects. Furthermore, we studied how a single  cause process influences a single effect process over time (there can be more than two processes, as is also the case in some of the presented examples, but the interventions were on a single process and the effect was detected in a single process). However, the proposed framework extends naturally so that time-windowed effects between multiple cause and effect processes can be analyzed. This might be relevant in some applications. Although we examined time-windowed causal effect matrices on linear systems only we framed the concepts for general Structural Causal Processes in which causal links are parameterized by possibly non-linear functions. Hence, future work could aim at finding approaches to estimate time-windowed effects in non-linear systems and investigate what causal-temporal features can be found therein.  

\paragraph{Acknowledgements} 
    This project receives funding from the European Union’s Horizon 2020 research and innovation programme under Marie Skłodowska-Curie grant agreement No 860100. Jakob Runge additionally was funded by the H2020 ERC Starting Grant CausalEarth (grant no. 948112)

\bibliography{refs}
\newpage
\null
\newpage
\appendix
\providecommand{\upGamma}{\Gamma}
\providecommand{\uppi}{\pi}
\section{Causal inference for temporal patterns in time-series: Supplementary Material}
\subsection{Graphical preliminaries}
Let $\mathcal{G} = (V\times \mathbb{Z}, E)$ be a causal time series graph. A path in $\mathcal{G}$ is given by sequence of nodes $p = (X^{i_1}(t_1), \cdots, X^{i_k}(t_k)$ such that either $X^{i_j}(t_j)\to X^{i_{j+1}}(t_{j+1})$ or $X^{i_{j+1}}(t_{j+1})\to X^{i_{j}}(t_{j})$ is a directed edge in $\mathcal{G}$ for every $j \in \{1, \cdots, t-1\}$. The node $X^{i_1}(t_1)$ is called starting point and $X^{i_k}(t_k)$ is called end point of $p$. A path $p$ from $X^{i_1}(t_1)$ to $X^{i_k}(t_k)$ is said to be causal if for every consecutive pair $X^{i_j}(t_j)\to X^{i_{j+1}}(t_{j+1})$ is an arrow in $\mathcal{G}$. Otherwise it is called non-causal. A non-causal path $q = (X^{i_1}(t_1), \dots,Y, \cdots, X^{i_k}(t_k) )$ is called a confounding path if $(Y, \cdots, X^{i_1}(t_1))$ is a causal path from $Y$ to $X^{i_1}(t_1)$ and $(Y, \cdots, X^{i_k}(t_k))$ is a causal path from $Y$ to $X^{i_k}(t_k)$. Let $\mathbf{X}$ and $\mathbf{Y}$ be disjoint sets of nodes in $\mathcal{G}$. A path from $\mathbf{X}$ to $\mathbf{Y}$ is called proper if only its starting point lies in $\mathbf{X}$. 

Let $\mathbf{X}$ be a set of nodes, then its set of parents $pa(\mathbf{X})$ consists of nodes $X'\in V \times \mathbb{Z}$ such that $X'\to X$ is an edge in $\mathcal{G}$, where $X \in \mathbf{X}$.  
The set of its ancestors $anc(\mathbf{X}) $ consists of vertices $X' \in \mathbf{X}$ such that there is a directed path $p$ starting at $X'$ and ending at $X\in \mathbf{X}$.
The set of its descendants $des(\mathbf{X})$ consists of vertices $X' \in \mathbf{X}$ such that there is a directed path $p$ emerging from some $X\in\mathbf{X}$ and landing in $X'$.

\subsection{Causal Effects in VAR processes}

A Vector Autoregressive Process (VAR) of order $p$ is given by a multivariate stochastic process $ \{\mathbb{X}_{t}\}_{t \in \mathbb{Z}}$, where for every $t \in \mathbb{Z}$, the $N$-dimensional random vector $\mathbb{X}(t)=(X^0(t), \dots, X^{n-1}(t), Y^{0}(t),\cdots, Y^{m-1}(t) )$ satisfies the following equation
\begin{equation}\label{eq:VAR}
\mathbb{X}(t) = \sum_{k=1}^{p} \Phi(p)\mathbb{X}(t) + \eta(t),
\end{equation} 
where $\Phi(p) \in \mathbb{R}^{N\times N}$ and $\eta(t)$ is Gaussian white noise.
Before proceeding with the analytic description of causal effects in VAR processes we need to introduce more notation. Let $Y^j(t)$ be a component in $\mathbb{X}(t)$ and $\mathbf{X} = \{X^{i_1}(s_1), \cdots, X^{i_n}(s_n) \} \subset V \times \mathbb{Z}$. For the description of causal effects in linear VAR processes it will be useful to have a notation for the following sets. 
First, the set of variables which are both ancestral to $Y(t)$ and lie in $\mathbf{X}$
\begin{equation}
anc(Y^j(t))_{\mathbf{X}} = anc(Y^j(t)) \cap \mathbf{X}. 
\end{equation}
Secondly, the set of nodes which are ancestral to $Y^j(t)$ but not descendants of $\mathbf{X}$
\begin{align}
anc(Y^j(t))_{\neg \mathbf{X}} =& anc(Y^j(t)) \setminus(\mathbf{X} \cup des(\mathbf{X})).
\end{align}
Let $p = (X^{i_1}(t_1), \cdots, X^{i_2}(t_k))$ be a causal path in $\mathcal{G}$, the causal time-series graph associated to the VAR process $\mathbb{X}$. Then the weight of $p$ is 
\begin{displaymath}
w(p) = \prod_{l= 1}^{k}\Phi(t_{l+1} - t_l)_{i_{l+1}, i_l}.
\end{displaymath}
Furthermore, denote by $\mathcal{P}(\mathbf{X}, X^{i_l}(s_l), Y^{i_k}(t_k))$ the set of all proper causal paths from $\mathbf{X}$ to $Y^j(t_j)$ with starting point $X^{i_l}(s_l)$.

\begin{proposition}\label{prop: linear effects var}
	Given a VAR and its associated causal time-series graph $\mathcal{G}$, two disjoint node sets $\mathbf{X}$ and $\mathbf{Y}$ in $V \times \mathbb{Z}$, the causal effect function of $\mathbf{X} = (X^{i_1}(s_1), \cdots, X^{i_n}(s_n))$ on $\mathbf{Y} = (Y^{i_1}(t_1), \cdots, Y^{i_m}(t_m))$. 
	In addition, each variable $Y^{i_j}(t_j)$ in $\mathbf{Y}$ admits the following description: 
	\begin{equation}\label{eq: time-series causal effect}
	Y^{i_j}(t_j) = \sum_{l=1}^{n} \theta(j, l) X^{i_l}(s_l) + g^j(\mathbf{X}'') + \eta,
	\end{equation}
	where $\mathbf{X}''$ is a random-vector which components are the elements of a finite subset of $anc(Y^{i_j}(t_j))_{\neg \mathbf{X}}$, $g^j$ is a linear function and $\eta$ is Gaussian random variable. 
	The coefficient
	\begin{displaymath}
	\theta(i,l) = \sum_{p \in \mathcal{P}(\mathbf{X}, X^{i_l}(s_l), Y^{i_j}(t_j))} w(p)
	\end{displaymath}
	is given by the sum over all proper weighted paths from $\mathbf{X}$ to $Y^{i_j}(t_j)$ with starting point $X^{i_l}(s_l)$. 
	In particular the causal effect function is linear and given by the following matrix
	\begin{displaymath}
	\Theta = (\theta(j,l))_{j,l} \in \mathbb{R}^{m\times n}.
	\end{displaymath}
	In this case the causal effect function will be referred to as the causal effect matrix. 
\end{proposition}

\begin{proof}
	Recursively resolve $Y^{i_j}(t_{i_j})$ as a linear combination of its ancestors and noise terms. Distinguish between those which are in $\mathbf{X}$ and those which are not. Only continue to resolve the ancestors which are not in $\mathbf{X}$. Stop the process once this sum only contains ancestors which are either in $\mathbf{X}$ or lie in the past relative to $\mathbf{X}$.   
\end{proof} 
\paragraph{Time-windowed causal effects in VAR processes}\label{sec: time windowed causal effect var}
Let $\X$ be a VAR process specified by equations of the form \ref{eq:VAR} and with corresponding CPG $\G$. We will now express the time windowed causal effect $\twce$ in terms of the VAR-matrices $\{\Phi(p)\}_{p\in \mathbb{Z}}$. By convention we set $\Phi(p)=\mathbf{0}$, if $p > \tuamax$ or $p \leq 0$. 
\begin{equation}
\mu(p) = 
\begin{cases}
\mathbb{I} & \text{, if } p = 0 \\
\sum_{k = 1}^{p} \mu(p-k) \cdot \Phi(k) & \text{, otherwise} 
\end{cases}
\end{equation}
This is the total causal effect matrix with lag $p$, i.e. $\mu(p)_{l,k}$ is the total causal effect $\X^k(t) \to \X^l(t+p)$, according to Proposition \ref{prop: linear effects var}. Further more we define a sequence of matrices $\hat{\Phi}(p)$, for $p \in \Z$ and which entries are defined as follows
\begin{equation}
(\hat{\Phi}(p))_{l,k} = 
\begin{cases}
\Phi(p)_{l,k} & \text{, if} k \neq i \\
0 & \text{, otherwise}
\end{cases} 
\end{equation}
And finally a bi-indexed family of matrices $\Psi(s,r) \in \mathbb{R}^{N\times N}$, where $s,r$ are both non-negative integers. They are defined according to the following inductive procedure
\begin{equation}
\Psi(s, r) = 
\begin{cases}
\mathbf{0} & \text{, if } s < 0 \\
\mathbb{I} & \text{, if } s = 0\\
\sum_{k = 1}^{s} \Psi(p-k) \cdot \hat{\Phi}(k) & \text{, if } 0<s \leq r \\
\mu(s-r) \cdot \Psi(r,r) & \text{, otherwise}
\end{cases}
\end{equation}
With the matrices $\Psi(\cdot, \cdot )$ at hand we can express the causal effect matrix $\twce$ corresponding to the time-windowed causal effect $\Xtr^i(t-\tau, T_i) \to \Xtr^j(t , T_j)$, i.e, the entries of $\twce$ are given as follows 
\begin{gather*}
(\twce)_{l,k} = \\
\begin{cases}
\Psi(T_i - k + \tau - T_j + l, T_i -k)_{j,i} & \text{, if } \\ &T_i -k + \tau - T_j + l\\& > 0 \\
0 & \text{, otherwise}
\end{cases}
\end{gather*}

\subsection{Granger Causality in the frequency domain}
Let $\X$ be VAR process with VAR-matrices $\{\Phi(p)\}_{1\leq p\leq\tuamax}$, whose noise terms are serially and mutually uncorrelated, then 
\begin{equation}
	\Phi(\omega) = \mathbb{I} - \sum_{p=0}^{\tuamax} e^{-i\omega k} \Phi(p)
\end{equation}
The transfer function of the VAR process is $\mathbf{H}(\omega) = \Phi(\omega)^{-1}$ and its cross spectral density (CSPD) matrix is $\mathbf{S}(\omega) = \mathbf{H}(\omega)\mathbf{H}(\omega)^H$, where $\mathbf{H}(\omega)^{H}$ the complex conjugate transpose. Then (unconditional) GC from process $i$ to process $j$ at frequency $\omega$ is  
\begin{equation}
	\mathcal{F}_{i\to j}^{GC}(\omega) = \log(\frac{\mathbf{S}_{j,j}(\omega)}{\mathbf{H}_{i,i}(\omega)\cdot\mathbf{H}_{i,i}(\omega)^{\ast}}),
\end{equation}  
where $\mathbf{S}_{k,l}(\omega)$ is the $(k,l)$-th entry of the CSPD matrix and $\mathbf{H}_{i,i}^{\ast}$ the complex conjugate of $\mathbf{H}_{i,i}(\omega)$.

\subsection{Imposing constraints on Causal Orthogonal Functions}
Note that the COFs do not account for the dynamics within $\X^i(t, T_i)$, since we deleted all arrows pointing to a node in $\Xtr^i(t, T_i)$. This poses the problem, that although signals in $\Imp$ may efficiently capture the information flow from $\Xtr^i(t-\tau, T_i) \to\Xtr^j(t, T_j)$, they do not inherently respect, whether these signals occur in a realization of $\X$. In that sense we might want the causal orthogonal functions to reside in the space of signals, which is best covered by observations. Alternatively, we might want to understand how $\X^j$ and $\X^i$ are interacting with respect to a specified pair of frequency bands, i.e., we might want to restrict the support of the power spectral density of the COFs. Finally, we might want the COFs to jointly adhere to a number of constraints. 
\paragraph{General Case}
Let $\mathbf{P} \in \mathbb{R}^{T_j \times k}$ and $\mathbf{Q}^{T_i \times l}$ be matrices with orthonormal columns, i.e. $\mathbf{P}^T \cdot \mathbf{P} = \mathbb{I} \in \mathbb{R}^{k \times k}$ and $\mathbf{Q}^T \cdot \mathbf{Q} = \mathbb{I} \in \mathbb{R}^{l\times l}$. Then we may form the following: 
\begin{displaymath}
\lambda' = \mathbf{Q}^T D\twce(\mathbf{x}) \mathbf{P} \in \mathbb{R}^{k \times l}.
\end{displaymath}
Per forming SVD on $\lambda'$ gives
\begin{displaymath}
\Sigma' = (\Resp')^T \lambda' \Imp',  
\end{displaymath}
Then, the COFs which respect the conditions encoded by $\mathbf{P}$ and $\mathbf{Q}$ are then given by the columns of the matrices
\begin{align*}
\Imp &= \mathbf{P} \cdot \Imp' \\
\Resp &= \mathbf{Q} \cdot \Resp' 
\end{align*}
Finally, assume we would like the COFs to satisfy several conditions simultaneously, encoded by the column-wise orthonormal matrices $\mathbf{P}_1 \in \mathbb{R}^{T_j \times k_1}, \mathbf{P}_2 \in \mathbb{R}^{T_j \times k_2}$ on the impulse side, and by $\mathbf{Q}_1\in \mathbb{R}^{T_i \times l_1}, \mathbf{Q}_2 \in \mathbb{R}^{T_i \times l_2}$ on the response side. Incorporating these conditions into the causal effect analysis of $\Xtr^i(t-\tau, T_i) \to \Xtr^j(t, T_j)$ yields the matrix
\begin{equation}
\lambda'' = \mathbf{Q}_2^T\mathbf{Q}_1\mathbf{Q}_1^T D\twce(\mathbf{x}) \mathbf{P}_1 \mathbf{P}_1^T \mathbf{P}_2
\end{equation}
Let the SVD of $\lambda''$ take the form 
\begin{equation}
\Sigma'' = (\Resp'')^T \lambda'' \Imp''.
\end{equation}
Hence, the COFs, constrained by the conditions $\mathbf{P}_1, \mathbf{P}_2$ on the impulse side and by $\mathbf{Q}_1, \mathbf{Q}_2$ on the response side, are given by 
\begin{align*}
\Imp &= \mathbf{P}_2\mathbf{P}_2^T\mathbf{P}_1 \Imp'' \\
\Resp &= \mathbf{Q}_2 \mathbf{Q}_2^T \mathbf{Q}_1 \Resp	
\end{align*} 
\paragraph{Investigate causal interactions between the dominant modes}
Another method to extract temporal features from a time-series is Singular Spectrum Analysis (SSA). This type of time-series analysis extracts temporal modes from the covariance matrix of the trajectory process $\Xtr^i$
\begin{equation}
\C_i(t, T) = \begin{pmatrix}
\Cov(X(t-T+i), X(t-T+j))
\end{pmatrix}_{i,j} \in \mathbb{R}^{T\times T}.
\end{equation}
As the covariance is symmetric, it is diagonalizable and therefore admits the following decomposition
\begin{equation}
\C_i(t,T) = \mathbf{E}_i^T \Xi \mathbf{E}_i, 
\end{equation}
where $\Xi$ is a diagonal matrix and $\mathbf{E}$ is an orthogonal matrix, which columns are uncorrelated temporal modes, explaining decreasingly much of the co-variance in $\X^i$ over a time-window of length $T$. When investigating cause-effect relations in a VAR process $\X$ in the form of $\X^i(t-\tau, T) \to X^j(t, T)$ we can use the features extracted with SSA in order to investigate causal interactions between the most dominant temporal patterns in $\Xtr^i(t-\tau, T)$ and $\Xtr^j(t, T)$, and finally detect new pairs of modes which are composed of just those and compactly represent the causal relationship
\begin{equation}
\mathcal{S} = (\mathbf{E}_i')^T D\twce(\mathbf{x}) \mathbf{E}_j',
\end{equation} 
here $\mathbf{E}_j'\in\mathbb{R}^{T\times k_j}$ and $\mathbf{E}_i' \in \mathbb{R}^{T\times k_i}$ contain a selection of columns of $\mathbf{E}_j$ resp. $\mathbf{E}_i$. These could, for example, be the modes which explain a certain percentage of the variance, thereby restricting the causal analysis to the spaces which will be best covered by data. Performing a SVD on $\mathcal{S}$ will yield COFs residing in these spaces. 
\begin{equation}
\mathcal{S} = (\Resp_i')^T \Sigma \Imp_j'.
\end{equation}
As already seen above, the COFs can now be determined from $\mathbf{E}_i'$ and $\Resp_i'$, resp. $\mathbf{E}_j'$ and $\Imp_j'$.
\end{document}